\documentclass[11pt]{amsart}
\usepackage{amsmath,amsthm,amscd,amsfonts,amssymb}
\usepackage[all]{xy}
\usepackage[ansinew]{inputenc}

\newtheorem{thm}{Theorem}[section]
\newtheorem{cor}[thm]{Corollary}

\theoremstyle{definition}
\newtheorem{defi}[thm]{Definition}

\newtheorem{nada}[thm]{}
\theoremstyle{remark}

\numberwithin{equation}{section}

\theoremstyle{remark}

\newcommand{\la}{\langle}
\newcommand{\ra}{\rangle}

\newcommand{\cC}{{\mathcal C}}

\renewcommand{\a}{\alpha}

\def\C{\mathcal {C}}

\begin{document}

\title[Relativistic forces in Lagrangian mechanics]
 {Relativistic forces in Lagrangian mechanics}

\author[J. Mu\~{n}oz D\'\i az]{J. Mu\~{n}oz D\'\i az}

\address{ Universidad de Salamanca, Departamento de Matem\'aticas, Plaza de la Merced 1, 37008-Salamanca (Spain)}

\email{clint@usal.es}

\begin{abstract}
We give a general definition of \emph{relativistic force} in the
context of Lagrangian mechanics. Once this is done we prove that
the only relativistic forces which are linear on the velocities
are those coming from differential 2-forms defined on the
configuration space. In this sense, electromagnetic fields provide
a mechanical system with the simplest type of relativistic forces.
\end{abstract}

\maketitle

\centerline{\today}

\vskip 1cm

\section{Introduction}

It remain in the folklore of Physics the Principles of
 Relativity (Special and General),  which Einstein would not have needed
 make, if he had worked in our time, when it is assumed that the fundamental
 laws should be formulated in a geometric, intrinsic way.

Part of the folklore is also the inconsistency of the ``Newton
Mechanics'' with the ``Relativity'', and it is considered no
necessary specify where the difference lies.

In this paper is given a precise definition of what a force
 relativistic is, and it is clarified how these forces are related to
 tensor fields.
It is a bit surprising that the relativistic nature of a
 force has nothing to do with the metric of the space of
 configuration.
Fundamental relativistic laws, as the Klein-Gordon equation, hold
with any metric, and then also with or without
 limitation of the speed of interaction (see \cite{RAlonso}).

\section{The geometry of the tangent bundle and the second order equations} \label{sec:notations}

Let $M$ be a smooth manifold of dimension $n$, and $TM$ be its
tangent bundle. Each differential $1$-form $\alpha$ on $M$ can be
considered as a function on $TM$, denoted by $\dot\alpha$, which
assigns to each $v_a\in T_a M$ the value
$$\dot\alpha(v_a)=\langle\alpha_a, v_a\rangle$$ obtained by duality.
In particular, a function $f\in\cC^\infty(M)$ defines the function
on $TM$ associated to $df$ that we denote in short by $\dot f$.
This definition also applies to differential forms $\alpha$ on
$TM$ that are at each point the pull-back of a form on $M$. In the
sequel we call these forms {\it horizontal forms}.

The map $f\mapsto \dot f$  from $\cC^\infty(M)$ to
$\cC^\infty(TM)$ is a derivation of the ring $\cC^\infty(M)$
taking values in the $\cC^\infty(M)$-module $\cC^\infty(TM)$. We
denote it by $\dot d$ since it is essentially the differential.
For each horizontal form $\a$, we have $\dot\a=\la\a, \dot d\ra$
as functions on $TM$.

Using the vector space structure of the fibers of $TM$ we can
associate to each $v_a\in T_aM$ a tangent vector to $T_aM$ at each
point as the derivative along $v_a$ in $T_aM$. Denoting by $V_a$
this derivation, we have for $f\in\cC^\infty(M)$ and a point $w_a\in
T_aM$:
\begin{equation*}
 V_a(\dot f)(w_a)=\lim_{t\rightarrow 0}
              \frac{\dot f(w_a+tv_a)-\dot f(w_a)}{t}=\dot
              f(v_a)=v_a(f)\ .
 \end{equation*}

At each $w_a\in T_aM$,  $V_a\in T_{w_a}(T_aM)$ is called the
\emph{vertical representative} of  $v_a\in T_aM$ and $v_a$ the
\emph{geometric representative} of $V_a$.

\begin{defi} \textbf{\em (Second Order Differential Equation).}\label{def: second_order_equation}
A vector field $D$ on $TM$ is a \emph{second order differential
equation} when its restriction (as derivation) to the subring
$\cC^\infty(M)$ of $\cC^\infty(TM)$ is $\dot d$. This is
equivalent to have $\pi_*(D_{v_a})=v_a$ at each point $v_a\in
T_aM$ (where $\pi\colon TM\to M$ denotes the canonical
projection).
\end{defi}

\begin{defi}\textbf{\em (Contact System).} \label{def: Pfaff_system}
The \emph{contact system} on $TM$ is the Pfaff system in $TM$ which
consists of all the $1$-forms annihilating all the second order
differential equations. It will be denoted by $\Omega$.
\end{defi}

The forms in the contact system also annihilate the
differences of second order differential equations, i.e. all
vertical fields. Therefore, they are horizontal forms; each
$\omega_{v_a}\in\Omega_{v_a}$ is the pull-back to $T^*_{v_a}TM$ of
a form in $T^*_aM$. Now, a horizontal $1$-form kills a second
order differential equation if and only if it kills the field
$\dot d$. Thus the contact system on $TM$ consists of the
horizontal $1$-forms which annihilate $\dot d$: a horizontal form $\alpha$ is contact if and only if $\dot\alpha=0$.
\medskip

\begin{nada}
\label{nada:local0} {\bf Local coordinate expressions.}
\end{nada}
We take local coordinates $q^i$, $i=1,\dots,n,$ in $M$ and
corresponding $q^i,\dot q^i$, in $TM$. We have, using Einstein
summation convention,

\begin{equation*}
\dot d=\dot q^i\frac{\partial}{\partial q^i}\ \ .
\end{equation*}

A vertical field has the expression
 \begin{equation*}
 V=f^i(q,\dot q)\frac{\partial}{\partial \dot q^i} \ \ ,
 \end{equation*}
and the one for a second order differential equation is
\begin{equation*}
D=\dot q^i\frac{\partial}{\partial q^i}+f^i(q,\dot
q)\frac{\partial}{\partial \dot q^i}\ \ .
\end{equation*}
Usually we will denote $f^i$ by $\ddot q^i$ understanding that it
is a given function of the $q$'s and $\dot q$'s.

A local system of generators for the contact system $\Omega$, out
of the zero section, is given by
\begin{equation*}
\dot q^idq^j-\dot q^jdq^i\quad (i,j=1,\dots,n)\ .
\end{equation*}

\bigskip

\section{Structure of a second order differential equation relative to a metric} \label{sec:second_order_equations}

Let $T^*M$ be the cotangent bundle of $M$ and $\pi\colon T^*M\to M$
the canonical projection. Recall that the \emph{Liouville form}
$\theta$ on $T^*M$ is defined by $\theta_{\alpha_a}=\pi^*(\alpha_a)$
for $\alpha_a\in T^*_aM$. Abusing the notation we can write
$\theta_{\alpha_a}=\alpha_a$.
\medskip

The $2$-form $\omega_2=d\theta$ is the natural symplectic form
associated to $T^* M$. In local coordinates $(q^1,\dots, q^n)$ in
$M$, and corresponding $(q^1,\dots, q^n,p_1,\dots, p_n)$ for
$T^*M$, we have

\begin{equation*}
\theta=p_idq^i,\qquad \omega_2=dp_i\wedge dq^i \ .
\end{equation*}

Let $T_2$ be a (non-degenerate) pseudo-Riemannian metric in $M$.
Then we have an isomorphism of vector fiber bundles
\begin{align*}
TM &\to T^*M\\
v_a &\mapsto i_{v_a}T_2
\end{align*}
($i_{v_a}T_2$ is the inner contraction of $v_a$ with $T_2$). Using
the above isomorphism we can transport to $TM$ all structures on
$T^*M$. In particular, we work with the Liouville form $\theta$
and the symplectic form $\omega_2$ transported in $TM$ with the
same notation.

From the definitions we have for the Liouville form in $TM$, at
each $v_a\in T_aM$,
\begin{equation}\label{Liouville}
\theta_{v_a}=i_{v_a}T_2\ ,
\end{equation}
where the form of the right hand side is to be understood
pulled-back from $M$ to $TM$.

\begin{defi}\textbf{\em (Kinetic Energy).}\label{cinetica}
The function $T=\frac 12\,\dot\theta$ on $TM$ is the kinetic
energy associated to the metric $T_2$. So, for each $v_a\in TM$,
we have $T(v_a)=\frac 12 \,\dot\theta(v_a)=\frac 12 T_2(v_a,v_a)$.
\end{defi}
Then, it can be proved the following main result
\begin{thm}[\cite{MecanicaMunoz}]\label{teoremaalfa}
The metric $T_2$ establishes a one-to-one correspondence between
second order differential equations and horizontal $1$-forms in
$TM$.

The second order differential equation $D$ and the horizontal
$1$-form $\a$ that correspond to each other are related by
\begin{equation}\label{formulaalpha}
i_D\omega_2+dT+\alpha=0\ .
\end{equation}
\end{thm}

The triple $(M,T_2, \alpha )$ will be called a \emph{mechanical
system}, where $M$ is the \emph{configuration space} which is
provided with a pseudo-Riemannian metric $T_2$ (to which
corresponds a \emph{kinetic energy} $T$ by (\ref{cinetica})), and
a \emph{work-form} or \emph{force form} $\alpha$ on $TM$. Finally,
correspondence (\ref{formulaalpha}) will be called the {\emph
Newton Law}: under the influence of a force $\alpha$ the
trajectories of the mechanical system satisfy (are the integral
curves of) the second orden differential equation $D$ associated with $\alpha$.

In particular, we consider the case of a mechanical system that is
undisturbed by any force.

\begin{defi}\textbf{\em (Geodesic Field).} \label{campogeodesico}
The geodesic field of the metric $T_2$ is the second order
differential equation, $D_G$, corresponding to $\alpha=0$:
\begin{equation}\label{ecuaciongeodesica}
i_{D_G}\omega_2+dT=0 \ .
\end{equation}
\end{defi}
The projection to $M$ of the curves solution of $D_G$ in $TM$ are
the geodesics of $T_2$. The geodesic field $D_G$ is chosen as the
origin in the affine bundle of second order differential
equations. With this choice we establish a one-to-one
correspondence between second order differential equations and
vertical tangent fields.
$$
D\longleftrightarrow W:=D-D_G.
$$
We define the covariant value of the second order differential
equation $D$, denoted by $D^\nabla$, as the field in $TM$ taking
values in $TM$ corresponding canonically to $W=D-D_G$.

\begin{nada}
\label{nada:local1} {\bf Local coordinate expressions.}
\end{nada}

Consider an open set of $M$ with coordinates $q^1,\dots,q^n$ and
the corresponding open set in $TM$ with coordinates
$q^1,\dots,q^n,\dot q^1,\dots,\dot q^n$. If the expression in
local coordinates of $T_2$ is
\begin{equation}\label{metrica}
T_2=g_{jk}\ dq^jdq^k
\end{equation}
then the local equations for the isomorphism $TM\approx T^*M$ are
\begin{equation}\label{isomorfismo}
p_j=g_{jk}\ \dot q^k.
\end{equation}
The Liouville form in $TM$ is given by
\begin{equation}\label{Liouvilletangente}
\theta=g_{jk}\ \dot q^kdq^j
\end{equation}
and for the kinetic energy we have, locally,
\begin{equation}\label{energiacinetica}
 T=\frac 12\, g_{ij}\dot q^i \dot q^j,\quad\text{so that}\quad p_j=\frac {\partial T}{\partial \dot q^j} \ .
\end{equation}
Let the second order differential equation $D$ be given by
\begin{equation}\label{formulaD}
D=\dot q^i\frac{\partial}{\partial q^i}+\ddot q^i
\frac{\partial}{\partial \dot q^i} \ ,
\end{equation}
where the $\ddot q^i$'s are given function of $q$'s and $\dot q$'s.

Now
\begin{equation}\label{alfa}
\alpha=-g_{lk}(\ddot q^l+\Gamma_{ij}^l\dot q^i\dot q^j) dq^k
\end{equation}
is the horizontal $1$-form related to $D$ by formula
(\ref{formulaalpha}). Equivalently,
\begin{equation}\label{alfa2}
\ddot q^l=-\left(g^{lk}\alpha_k+\Gamma_{ij}^l\dot q^i\dot
q^j)\right)
\end{equation}
In particular, for the geodesic field we have,
\begin{equation}\label{geodesico}
D_G=\dot q^i\frac{\partial}{\partial q^i}-\Gamma_{ij}^l \dot q^i
\dot q^j \frac{\partial}{\partial \dot q^l}
\end{equation}
and, finally, the covariant value of $D$ is
\begin{equation}\label{coordenadascovariante}
 D^\nabla=(\ddot q^l+\Gamma_{ij}^l\dot q^i\dot q^j)
              \frac{\partial}{\partial q^l}=
              -g^{lk}\alpha_k
              \frac{\partial}{\partial q^l}\ .
\end{equation}

\section{Relativistic forces}

For physicists, the motion of each point particle in Relativity
(yet in the special one) is parametrized by the ``proper time'' of
that particle, whose ``infinitesimal element'' $ds$ is the so
called lenght element associated with $T_2$; thus, in Special
Relativity, it is written
$$ds^2=dt^2-dx^2-dy^2-dz^2=(dx^0)^2-\sum_1^3(dx^i)^2$$
(by taking units such that the velocity of the light is $c=1$).

For any metric, the length element is $\theta/\sqrt{\dot\theta}$,
where $\theta$ is the Liouville form. In this way,  the
classical $ds$ is the restriction of such a length element
to the curve in $TM$ describing the lifting of the corresponding
parametrized curve in $M$. It turns out that length does not
depend on the parametrization, and this is the reason for which it
has sense  to talking about the length of a curve.

If $D$ is a second order differential equation on $M$, when we say
that a curve solution can be parametrized  by the length element
$ds$ we means that the proper parameter for such a curve solution
of $D$ is the specialization of $\theta/\sqrt{\dot\theta}$; that
is to say,
$\theta(D)/\sqrt{\dot\theta}=\dot\theta/\sqrt{\dot\theta}=1$.
Thus, $$\dot\theta=0\quad\text{ or}\quad \dot\theta=\pm 1$$ on
such a curve solution of $D$.

Therefore, the second order differential equations on $M$ which
describe relativistic motions are vector fields $D$ (on manifold
$TM$) tangent to the hypersurfaces $\dot\theta=0$, $\dot\theta=1$,
and  $\dot\theta=-1$. Let us put $D=D_G+W$. From the
Newton-Lagrange equation $i_D\omega_2+dT+\alpha=0$ we get
$$0=i_Di_D\omega_2+i_D dT+i_D\alpha=DT+\dot\alpha$$
on such submanifolds. As a consequence,
\begin{equation}\label{EcuacionT}
WT+\dot\alpha=0
\end{equation}
when $\dot\theta=0,\pm 1$. On the other hand, field $W$ is
vertical and function $T$ is homogeneous of second degree, so that
$$WT=\langle\theta,w\rangle,$$
where $w$ is the geometric representative of $W$ (i.e.,
$i_W\omega_2=i_wT_2$). In local coordinates,
\begin{equation}\label{coordenadasA}
\langle\theta,w\rangle=g_{ij}\,\dot q^i w^j,
\end{equation}
 where
$w=w^i\partial/\partial q^i$. If the $w^i$'s are homogeneous with
respect to dotted coordinates $\dot q^i$, expression
(\ref{coordenadasA}) vanishes for $\dot\theta=1$ if and only
vanishes for all $\dot\theta$.

Therefore, `relativistic' fields $D=D_G+W$ such that $W$ is
homogeneous with respect to velocities, have to be tangent to each
manifold $\dot\theta=k$, $k\in \mathbb{R}$. That is to say,
$\dot\theta=2T$ is a first integral of $D$ and then, also of $W$.
Since Equation (\ref{EcuacionT}) we derive $\dot\alpha=0$. In
other words, $\alpha$ is a contact differential form.

On the other hand, each second order differential equation $D$ on
$M$ can be modified by means of a ``relativistic constraint''
(analogous to the time constraints, see \cite{MecanicaMunoz}) in
order to be tangent to manifolds $\dot\theta=\text{constant}$. In
the case of $D$ being tangent to $\dot\theta=0$, $\dot\theta=1$,
$\dot\theta=-1$, on these hypersurfaces does not suffer any
modification. For all the above, it seems reasonable the following

\begin{defi}\label{DefRelativista}
A \emph{relativistic field} on $(M,T_2)$ is a second order
differential equation $D$ such that $DT=0$, where $T$ is the
kinetic energy associated with $T_2$.
\end{defi}

The previous discussion gives us
\begin{thm}\label{criterio}
A second order differential equation $D$ on $(M,T_2)$ is a
relativistic field if and only if the work form $\alpha$
associated by virtue of the Newton-Lagrange law
($i_D\omega_2+dT+\alpha=0$) belongs to the contact system $\Omega$
of $M$.
\end{thm}

\begin{cor}\label{corolario1}
The differential 1-forms $\alpha$ corresponding, as work forms, to
relativistic fields, are the same, independently of the metric
$T_2$: they are those belonging to the contact system $\Omega$ of
$TM$.
\end{cor}

\begin{cor}\label{colorario2}
A mechanical system $(M,T_2,\alpha)$ is relativistic if and only
if, for each parametrized curve on $M$, solution of the system,
the tangent field $u$ has constant length (that is to say,
$T_2(u,u)=\|u\|^2$ is constant along the curve).
\end{cor}
\begin{proof}
We have $-\textrm{grad}(u)=D^\nabla=u^\nabla u$ along the given
curve. By taking scalar product by $u$, it holds $T_2( u^\nabla u
,u)=-\alpha(u)=-\dot\alpha(u)$;  we can arbitrarily to fix a
tangent vector in $M$ as an initial condition for a second order
differential equation; thus, we derive that $\dot\alpha=0$
$\Leftrightarrow$
 $T_2(u^\nabla u,u)=0$ for all curve solution of $D$.
 Then the  result follows from the identity
 $T_2( u^\nabla u,u)= (1/2) u\left(T_2(u,u)\right)$.
 \end{proof}

 In the sense specified in Corollary \ref{colorario2},
 relativistic fields are the natural generalization of the geodesic
 field $D_G$, the one corresponding with $\alpha=0$ (in the
 contact system).

 At first glance, it is noteworthy that relativistic forces are characterized
 by corresponding to forms of work that are defined before any metric (the contact
 system of $TM$) and, on the other hand, given the metric, these forces are characterized
 by having as solutions curves whose tangent vector is of constant length, or, equivalently, by
 preserving the kinetic energy associated to the metric.

In local coordinates $q^1,\dots,q^n$ for $M$, a local basis of the
contact system in $TM$ is the comprised by the forms $\dot
q^i\,dq^j-\dot q^j\,dq^i$, ($i<j$). The 2-form $A_{ij}:=dq^i\wedge
dq^j$ holds
$$i_{\dot d}A_{ij}=\dot q^i\,dq^j-\dot q^j\,dq^i.$$
Therefore, for each relativistic force, its 1-form of work
$\alpha$ is $i_{\dot d}\Phi_2$ where $\Phi_2$ is a ``horizontal''
2-form, linear combination of the $A_{ij}$' with coefficients in
$\C^\infty(TM)$. That is to say, $\Phi_2$ is a field on $TM$ with
values in $\Lambda_2M$ (the fiber bundle of the differential
2-forms). Or, which is the same, $\Phi_2$ is a section of
$TM\times_M\Lambda_2M\to TM$. For each contact 1-form $\alpha$,
there exists a tensor field $\Phi_2$ with $i_{\dot
d}\Phi_2=\alpha$; such a $\Phi_2$ is not unique: it is
undetermined up to a sum of forms $\phi^i\alpha_i\wedge\beta_i$
where the $\alpha_i$, $\beta_i$ are contact forms, and
$\phi^i\in\C^\infty(TM)$.

Let $(M,T_2,\alpha)$ be a mechanical system whose equation of
motion $D=D_G+W$ corresponds with a force (vertical tangent field
in $TM$) $W$ which depends linearly on the velocities. This means
that, at the point $(a,v_a)\in T_aM$, the vector $w_a$ is the
transformed of $v_a$  through a linear transformation
$\Phi_a\in\textrm{Hom}(T_aM,T_aM)$. Letting the point $a$ run
along $M$ we will have a tensor field $\Phi$ on $M$ (a field of
endomorphisms) such that $w=\Phi(\dot d)$; then, $W=\Phi(V)$,
where $V$ is the tangent field associated with $\dot d$ ($V=\dot
q^i\partial/\partial \dot q^i$, in local coordinates). The work
form $\alpha$ and the force $W$ are related be means of the
Newton-Lagrange equation $i_W\omega_2+\alpha=0$ or
$i_wT_2+\alpha=0$.

If $\alpha$ is a relativistic force, then $\alpha=i_{\dot
d}\Phi_2$ for a differential 2-form on $TM$ with values in
$\Lambda^2M$ (for short, a \emph{horizontal 2-form}). Thus,
$i_{\Phi(\dot d)}T_2+i_{\dot d}\Phi_2=0$; by contracting with
$\dot d$ it gives us $T_2(\Phi(\dot d),\dot d)=0$, which is
equivalent to $\Phi$ being the field endomorphism associated with
an alternate tensor field:
$$\Phi(\dot d)=-\textrm{grad} i_{\dot d}\Phi_2.$$

Note that for any 2-form $\Phi_2$ on $M$, the contact 2-form
$i_{\dot d}\Phi_2$ completely determines $\Phi_2$. As a
consequence
\begin{thm}\label{caracterizacion}
The forces (vertical vector fields on $TM$) which are relativistic
and linearly depend on the velocities are canonically associated
with 2-forms on $M$, once the metric $T_2$ is given on $M$.
\end{thm}
\bigskip

 \section*{Aknowledgements}
 This note is a result of a work in collaboration with my dear
 friend and colleague, Prof. R. Alonso.
 \bigskip

\renewcommand*{\refname}{Bibliografía}

\end{document}